\documentclass[runningheads]{llncs}
\usepackage{graphicx}
\usepackage[T1]{fontenc}
\usepackage{lmodern}
\usepackage[utf8]{inputenc}
\usepackage{fullpage} 

\usepackage{hyperref}
\usepackage{csquotes}
\usepackage{microtype}
\usepackage{mathtools}
\usepackage{amsmath}
\usepackage{amssymb}
\usepackage{booktabs}
\usepackage{multirow}
\usepackage{wrapfig}
\usepackage[margin,noinline,draft]{fixme}
\usepackage{cryptocode}
\usepackage{tikz}
\usepackage{circuitikz}
\usetikzlibrary{matrix}
\usetikzlibrary{graphs}
\usetikzlibrary{positioning, arrows.meta, calc}

\usepackage{tabularx}

\newcommand{\Z}{\mathbb{Z}}

\DeclareMathOperator{\enc}{enc}

\DeclareMathOperator{\td}{td}

\newcommand{\df}{\mathrel{\mathop:}=}
\newcommand{\fd}{=\mathrel{\mathop:}}

\newcommand{\A}{\mathcal{A}}

\newcommand{\svdots}{\raisebox{3pt}{$\scalebox{.6}{$\vdots$}$}}
\newcommand{\sdots}{\raisebox{3pt}{$\scalebox{.6}{$\dots$}$}}
\newcommand{\sddots}{\raisebox{3pt}{$\scalebox{.6}{$\ddots$}$}}

\usepackage{xspace}
\usepackage[sort&compress,numbers]{natbib}

\begin{document}

\fontsize{11}{13}\selectfont

\title{Tight Lower Bounds for Block-Structured Integer Programs\thanks{C. Hunkenschröder acknowledges funding by Einstein Foundation Berlin. K.-M. Klein  was supported by DFG project KL 3408/1-1. A. Lassota  was partially supported by the Swiss National Science Foundation (SNSF) within the project \emph{Complexity of Integer Programming (207365)}. M. Kouteck{\'{y}} was partially supported by the Charles University project UNCE 24/SCI/008 and by the project 22-22997S of GA ČR. A. Levin is partially supported by ISF -- Israel Science Foundation grant number 1467/22.\\ An extended abstract version of this work appreared in the Proceedings of IPCO 2024.}}

\authorrunning{Hunkenschröder et al.}

\author{Christoph Hunkenschröder\inst{1}\orcidID{0000-0001-5580-3677}\and
Kim-Manuel Klein \inst{2}\orcidID{0000-0002-0188-9492}\and
Martin Kouteck{\'{y}} \inst{3}\orcidID{0000-0002-7846-0053}\and
Alexandra Lassota \inst{4}\orcidID{0000-0001-6215-066X}\and
Asaf Levin\inst{5}\orcidID{0000-0001-7935-6218}}

\institute{(formerly) Institut für Mathematik, TU Berlin, Berlin, Germany\\
\email{hunkenschroeder@tu-berlin.de}\\ \and
Institute for Theoretical Computer Science, University of Lübeck, Lübeck, Germany\\
\email{kimmanuel.klein@uni-luebeck.de}\\ \and
Computer Science Institute, Charles University, Prague, Czech Republic\\
\email{koutecky@iuuk.mff.cuni.cz}\\ \and
Eindhoven University of Technology, Eindhoven, The Netherlands\\
\email{a.a.lassota@tue.nl}\and
Faculty of Data and Decisions Sciences, The Technion, Haifa, Israel\\
\email{levinas@ie.technion.ac.il}\\}

\maketitle

\begin{abstract}
We study fundamental block-structured integer programs called tree-fold and multi-stage IPs. Tree-fold IPs admit a constraint matrix with independent blocks linked together by few constraints in a recursive pattern; and transposing their constraint matrix yields multi-stage IPs. 
The state-of-the-art algorithms to solve these IPs have an exponential gap in their running times, making it natural to ask whether this gap is inherent.
We answer this question affirmative. Assuming the Exponential Time Hypothesis, we prove lower bounds showing that the exponential difference is necessary, and that the known algorithms are near optimal.
Moreover, we prove unconditional lower bounds on the norms of the Graver basis, a fundamental building block of all known algorithms to solve these IPs.
This shows that none of the current approaches can be improved beyond this bound.
\end{abstract}

\section{Introduction}
\label{sec:intro}
In the past years, there has been tremendous progress in the algorithmic theory and in the applications of \emph{block-structured integer programming}. An \emph{integer program~(IP)} in standard form is the problem $\min \{c^\intercal x: \ Ax = b,\, l \leq x \leq u,\, x \in \Z^n\}$. We deal with the setting where the constraint matrix $A$ exhibits a certain block structure.
One of the most prominent block-structure are \emph{$n$-fold} integer programs~($n$-fold IPs), in which the constraint matrix~$A$ decomposes into a block-diagonal matrix after the first few rows are deleted.
In other words, $n$-fold IPs are constructed from independent IPs of small dimensions that are linked by a few constraints.
The generalization, in which the diagonal blocks themselves have an $n$-fold structure recursively, is called \emph{tree-fold} IPs.
The transpose of a tree-fold matrix yields another class of highly relevant constraint matrices called \emph{multi-stage} matrix. We formally define those next.

\begin{definition}[Tree-fold and multi-stage matrices]
Any matrix $A \in \Z^{m \times n}$ is a tree-fold matrix with one level and level size $m$.
A matrix $A$ is a tree-fold matrix with $\tau \geq 2$ levels and level sizes $\sigma = (\sigma_1, \dots, \sigma_\tau) \in \Z_{\geq 1}^{\tau}$
if deleting the first $\sigma_1$ rows of $A$ decomposes the matrix into a block-diagonal matrix, where each block is a tree-fold matrix with $\tau-1$ levels and level sizes~$(\sigma_2, \dots, \sigma_\tau)$.

If $A^\intercal$ is a tree-fold matrix with $\tau$ levels and level sizes $\sigma$, then $A$ is called a \emph{multi-stage matrix} with $\tau$ stages and stage sizes $\sigma$.
\end{definition}

For a schematic picture, see Figure~\ref{fig:treefoldmultistage}
   
\begin{figure}
\centering
\resizebox{0.8\textwidth}{!}{
\begin{circuitikz}
\tikzstyle{every node}=[font=\LARGE]
\draw [, line width=0.5pt ] (2.5,12.25) rectangle (3.25,4.75);
\draw [, line width=0.5pt ] (3.5,12.25) rectangle (5.25,8.5);
\draw [, line width=0.5pt ] (5.5,12.25) rectangle (6.5,11);
\draw [, line width=0.5pt ] (6.75,11) rectangle (7.75,9.75);
\draw [, line width=0.5pt ] (8,9.75) rectangle (9,8.5);
\draw [, line width=0.5pt ] (9.25,8.5) rectangle (11.25,4.75);
\draw [, line width=0.5pt ] (11.5,8.5) rectangle (12.5,7.25);
\draw [, line width=0.5pt ] (12.75,7.25) rectangle (13.75,6);
\draw (14,6) rectangle (15,4.75);
\draw [line width=1pt, short] (2.5,12.5) -- (2,12.5);
\draw [line width=1pt, short] (2,12.5) -- (2,4.5);
\draw [line width=1pt, short] (2,4.5) -- (2.5,4.5);
\draw [line width=1pt, short] (14.75,4.5) -- (15.25,4.5);
\draw [line width=1pt, short] (15.25,4.5) -- (15.25,12.5);
\draw [line width=1pt, short] (15.25,12.5) -- (14.75,12.5);
\draw [, line width=0.5pt ] (20,12.25) rectangle (30,11.25);
\draw [, line width=0.5pt ] (20,11) rectangle (25,9.75);
\draw [, line width=0.5pt ] (20,9.5) rectangle (22.5,9.25);
\draw [, line width=0.5pt ] (20,9) rectangle (21.25,8.5);
\draw [, line width=0.5pt ] (21.25,8.25) rectangle (22.5,7.75);
\draw [, line width=0.5pt ] (22.5,7.5) rectangle (25,7.25);
\draw [, line width=0.5pt ] (22.5,7) rectangle (23.75,6.5);
\draw [, line width=0.5pt ] (23.75,6.25) rectangle (25,5.75);
\draw [, line width=0.5pt ] (25,5.5) rectangle (30,4.25);
\draw [, line width=0.5pt ] (25,4) rectangle (27.5,3.75);
\draw [, line width=0.5pt ] (25,3.5) rectangle (26.25,3);
\draw [, line width=0.5pt ] (26.25,2.75) rectangle (27.5,2.25);
\draw [, line width=0.5pt ] (27.5,2) rectangle (30,1.75);
\draw [, line width=0.5pt ] (27.5,1.5) rectangle (28.75,1);
\draw [, line width=0.5pt ] (28.75,0.75) rectangle (30,0.25);
\draw [line width=1pt, short] (20,12.5) -- (19.5,12.5);
\draw [line width=1pt, short] (19.5,12.5) -- (19.5,0);
\draw [line width=1pt, short] (19.5,0) -- (20,0);
\draw [line width=1pt, short] (29.75,0) -- (30.25,0);
\draw [line width=1pt, short] (30.25,0) -- (30.25,12.5);
\draw [line width=1pt, short] (30.25,12.5) -- (29.75,12.5);
\end{circuitikz}
}
\caption{On the left, a schematic multi-stage with three levels is presented. On the right, a schematic tree-fold with 4 layers is pictured. All entries within a rectangle can be non-zero, all entries outside of the rectangles must be zero.}\label{fig:treefoldmultistage}
\end{figure}
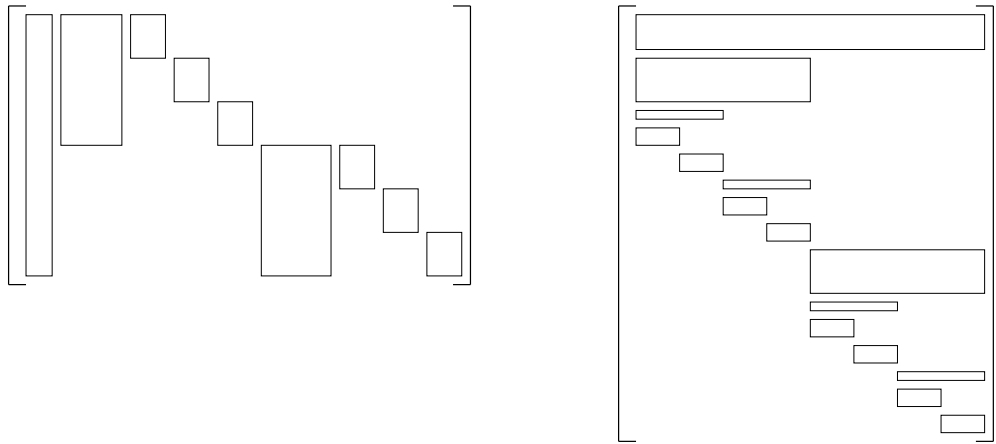

\begin{definition}[$n$-fold and $2$-stage stochastic matrices]
The special case of a tree-fold matrix with two levels is called $n$-fold. Respectively, a multi-stage matrix with two stages is called $2$-stage stochastic matrix.
\end{definition}

For a picture, see Figure~\ref{fig:nfoldtwostage}

\begin{figure}
\centering
\begin{tikzpicture}
    \node at (0,0) {$\begin{bmatrix}
        A_1 & D_1 & & & \\
        A_2 & & D_2 & & \\
        \vdots & & & \ddots & \\
        A_n & & & & D_n
    \end{bmatrix}$};

    \node at (6,0) {$\begin{bmatrix}
        C_1 & C_2  & \dots & C_n \\
        D_1 & & & \\
        & D_2 & & \\
        & & \ddots & \\
        & & & D_n
    \end{bmatrix}$};
\end{tikzpicture}
\caption{On the left, a $2$-stage stochastic matrix with blocks $A_i$ and $D_i$, $i \in [n]$ is presented. On the right, an $n$-fold matrix with blocks $C_i$ and $D_i$, $i \in [n]$ is pictured. All entries not belonging to a block are zero.}\label{fig:nfoldtwostage}
\end{figure}

The study of $n$-fold IPs was initiated in~\cite{DBLP:journals/disopt/LoeraHOW08}. A milestone was the fixed-parameter tractable algorithm by Hemmecke et al.~\cite{hemmecke2013n} whose running time depends polynomially on the dimension $n$, and exponentially only on the sizes of the small blocks.
Faster and more generally applicable algorithms were subsequently developed, including strongly polynomial algorithms, near-linear (in $n$) time algorithms~\cite{DBLP:journals/disopt/LoeraHOW08,DBLP:conf/icalp/EisenbrandHK18,koutecky,DBLP:journals/siamdm/JansenLR20,DBLP:conf/soda/CslovjecsekEHRW21,eisenbrand2019algorithmic}, and a new result where entries in the global part can be large~\cite{cslovjecsek2023parameterized}.
At the same time,
$n$-fold IPs found many applications,
for instance in scheduling problems,
stringology, graph problems, and computational social choice, see e.g.~\cite{DBLP:journals/disopt/GavenciakKK22,DBLP:journals/ol/HemmeckeOW11,DBLP:journals/mp/KnopKM20,maack,DBLP:conf/aips/JansenLMP21,DBLP:journals/scheduling/KnopK18,DBLP:conf/esa/KnopK22,DBLP:journals/orl/KnopKLMO21,DBLP:journals/teco/KnopKM20,DBLP:conf/stacs/0011MYZ17},
solving long-standing open problems.

Multi-stage IPs and their special case for $\tau = 2$ (called $2$-stage stochastic IPs) have been studied even longer than $n$-fold IPs, going back to the work of Aschenbrenner and Hemmecke~\cite{AH}.
They are commonly used in stochastic programming and often used in practice to model uncertainty of decision making over time~\cite{DBLP:journals/eor/Albareda-SambolaVF06,DBLP:journals/mor/DempsterFJLLK83,kall1994stochastic,DBLP:journals/ior/LaporteLM94}.
The first known upper bounds on the complexity of solving multi-stage IPs had a huge and non-explicit dependence on $\tau$ and $\sigma$ in their running time.
The upper bound was subsequently improved to have an exponential tower of height~$\tau$ with $\|\sigma\|_1$ appearing at the top (times a polynomial in the encoding length of the input $I$), and only very recently to have a triple-exponential (in $\|\sigma\|_1$) running time (times $|I|^{O(1)}$)~\cite{DBLP:conf/esa/CslovjecsekEPVW21,eisenbrand2019algorithmic,DBLP:journals/mp/Klein22,klein2021collapsing}.
In~\cite{cslovjecsek2023parameterized}, a new FPT time algorithm to decide feasibility of a $2$-stage stochastic IP is presented that can also handle large entries in the global part (i.e., the largest entry in the global part is not a parameter). 

Intriguingly, the algorithms for tree-fold IPs have a running time that depends only double-exponentially on $\| \sigma \|_1$~\cite{eisenbrand2019algorithmic}.
A natural response to seeing this exponential gap is to ask whether it is inherent, and whether multi-stage integer programming is indeed harder than tree-fold IPs despite their similar nature of constraint matrices.
This question was partially answered when Jansen et al.~\cite{jansen2021double} showed that, assuming the Exponential Time Hypothesis (ETH), $2$-stage stochastic IPs require a double-exponential running time in $\|\sigma\|_1$. This contrasts with known single-exponential algorithms w.r.t. $\|\sigma\|_1$ for $n$-fold IPs.
Also for tree-fold IPs, a double-exponential lower bound w.r.t. $\tau$ is known~\cite{knop2020tight}, although this is stated in terms of the parameter \emph{tree-depth}, which is linked to the largest number of non-zeroes in any column (see~Section~\ref{sec:tree-depth}), still leaving the exact dependence on the number of levels $\tau$ open.

We settle the complexity of all aforementioned block-structured integer programs, answer the question whether the exponential gap is necessary affirmatively, and (nearly) close the gaps between algorithms and lower bounds:
\begin{enumerate}
\item (\textbf{Theorem~\ref{thm:multi-stage}}) We show an ETH-based lower bound for multi-stage IPs that is triple-exponential in the number of levels $\tau$ when the level sizes $\sigma$ are constant, and recovers the existing double-exponential lower bound (in $\|\sigma\|_1$)~\cite{jansen2021double} as a special case.
This lower bound is comparable to the running time of the currently fastest algorithm~\cite{klein2021collapsing}.
\item (\textbf{Theorem~\ref{thm:tree-fold}}) We show an ETH-based lower bound for tree-fold IPs which recovers as a special case the result of~\cite{knop2020tight}, and is comparable to the running time of the currently best algorithm~\cite{eisenbrand2019algorithmic}. Our bound shows more accurately how the running time depends on $\tau$.
\item (\textbf{Corollary~\ref{cor:n-fold}}) A particularly interesting consequence of Theorem~\ref{thm:tree-fold} is a lower bound of roughly $2^{\sigma_1 \sigma_2}$ for the special class of $n$-fold IPs.
\end{enumerate}

The core technical idea behind the lower bounds in this paper relates bi- and tri-diagonal matrices to block-structured matrices, see Section~\ref{sec:BiTri}.
\begin{enumerate}
\item[4.] (\textbf{Lemma~\ref{lem:diagonal}} and \textbf{Corollary~\ref{cor:diagonal}}) Every bi-/tri-diagonal matrix can be reordered to obtain a multi-stage (Lemma~\ref{lem:diagonal}) and a tree-fold (Corollary~\ref{cor:diagonal}) matrix.
\end{enumerate}
We believe this result is of independent interest, as it provides a new tool for solving combinatorial problems: Formulating any problem as a matrix with constant bandwidth is enough to be able to apply the tree-fold or multi-stage integer programming algorithms to solve it efficiently.

Since the hard instances we construct have bi- or tri-diagonal structure, we are able to obtain the required lower bounds.
For the multi-stage IPs lower bound, we combine this idea with splitting one complicated constraint carefully into several sparse constraints with only few variables, similarly to the well-known reduction from a \textsc{$3$-Sat} formula to a \textsc{$3$-Sat} formula where each variable only appears constantly often.
This is done in Section~\ref{sec:MultiStage}.
Section~\ref{sec:TreeFold} continues with the lower bound for tree-fold IPs.

The central concept to all the aforementioned algorithms is the \emph{Graver basis}~$\mathcal{G}(A)$ of the constraint matrix $A$, or in case of the new result~\cite{cslovjecsek2023parameterized}, of some reduced constraint matrix $A'$ with block-structure and small entries, and a closer examination shows that the main factor driving the complexity of those algorithms are the $\ell_\infty$- and $\ell_1$-norms of elements of $\mathcal{G}(A)$ ($\mathcal{G}(A')$).
Improved bounds on those norms would immediately lead to improved algorithms, contradicting ETH.
We show that the ETH cannot be violated in this way by giving \emph{unconditional} lower bounds on the norms of elements of $\mathcal{G}(A)$ for block-structured matrices.

We demonstrate these norm lower bounds on the matrices used in the proofs of Theorems~\ref{thm:multi-stage} and~\ref{thm:tree-fold}, as described in Section~\ref{sec:Graver}. Our unconditional lower bounds for Graver basis elements traces back to the same matrix hardness was proven for in~\cite{DBLP:journals/mp/Klein22}. Extending these results to multi-stage IPs and advancing to tree-fold and $n$-fold IPs though required the here presented concept of rearranging bi- and tri-diagonal matrices.
In Section~\ref{sec:tree-depth}, we briefly express our work in terms of the parameter \emph{tree-depth} that is also commonly used to describe block-structured integer programs.

This paper partially builts on an arXiv preprint~\cite{eisenbrand2019algorithmic}.
The present paper provides stronger and novel results. Specifically, it introduces bi- and tri-diagonal matrices formally as crucial components for the hardness proofs. This approach enables us to reframe the results for tree-fold IPs in terms of \emph{levels} and the maximum number of rows in a level, rather than stating them solely on tree-depth as in both~\cite{eisenbrand2019algorithmic,knop2020tight}, which is the product of these parameters. 
This refined perspective allows for a more nuanced analysis of $n$-fold IPs which derives at its (near tight) lower bound, which was unattainable with previous methods.
In the context of multi-stage IPs, our investigation picks up from where the proof in~\cite{jansen2021double} for $2$-stage stochastic IPs concluded. Notably, the proof for $2$-stage stochastic IPs did not involve or observe any potential rearrangement of the stages; the first stage comprised only one variable, and the $2$-stage stochastic structure emerged naturally from the underlying problem.
 
\section{About Bi-Diagonal and Tri-Diagonal Matrices}
\label{sec:BiTri}
This section is devoted to bi- and tri-diagonal matrices, i.e., matrices in which all non-zero entries are are on two or three consecutive diagonals, respectively.
\begin{definition}[Bi-diagonal, Tri-diagonal Matrix]
A matrix $\A \in \Z^{m \times n}$ is \emph{bi-diagonal} if $a_{i,j} = 0$ for $i \notin \{j,j+1\}$.
A matrix $\A \in \Z^{m \times n}$ is \emph{tri-diagonal} if $a_{i,j} = 0$ for $i \notin \{j,j+1,j+2\}$.
\end{definition}
We show that any bi- or tri-diagonal matrix $A$ can be viewed as a multi-stage or tree-fold matrix with the desired parameters.
While the next lemma requires quite specific matrix dimensions, note that once a matrix is bi- or tri-diagonal, we can always add zero rows or columns, and it remains bi- or tri-diagonal.

\begin{lemma}
\label{lem:diagonal}
Let $\tau \geq 1$, $\sigma \in \Z_{\geq 1}^\tau$, and define $S \df \prod_{i=1}^\tau (\sigma_i + 1)$.
\begin{itemize}
\item[$i)$] Let $\A \in \Z^{ S \times S - 1}$ be bi-diagonal.
Then $\A$ is a multi-stage matrix with $\tau$ stages and stage sizes $\sigma$, up to column permutations.
\item[$ii)$] Let $\A \in \Z^{2S \times 2S - 2}$ be tri-diagonal.
Then $\A$ is a multi-stage matrix with $\tau$ stages and stage sizes $2\sigma$, up to column permutations.
\end{itemize}
\end{lemma}

\paragraph*{Proof idea.}
    We prove both claims by induction on the number of stages~$\tau$. The claims trivially hold for $\tau = 1$ as any matrix can be interpreted as a multi-stage matrix of just one stage. 

    i) As for bi-diagonal matrices, note that if we delete a column $i$, the matrix separates into two independent bi-diagonal matrices, one from column $1$ to column $i-1$, and the second from columns $(i+1)$ until $(S-1)$ (and the corresponding non-zero rows respectively).

    We use this idea of splitting the matrix into independent matrices as follows: for $\tau > 1$, we permute $\sigma_1$ equidistant columns to the front. This gives us a $2$-stage stochastic matrix with $\sigma_1$ columns in the first stage and a second stage that are exactly the independent matrices. For each of these independent matrices, we now need to find a re-arrangement into a multi-stage matrix with $\tau-1$ stages and sizes $\sigma_2, \dots, \sigma_\tau$, which is possible due to the induction hypothesis. 

    ii) The only adaption needed for the second claim is to shift two consecutive columns for each column chosen in i) to the front to split the matrix into independent matrices. 

\begin{proof}
$i):$
The proof is by induction on $\tau$, and the base case for $\tau=1$ is trivial.
For $\tau > 1$, we briefly lay out the idea before providing the formal construction.
Observe that if we delete any column, we can partition the remaining matrix into two blocks $\binom{A_1}{0}$, $\binom{0}{A_2}$ whose columns are orthogonal to each other as depicted below for column $j$:
\[
\A = \left(
\begin{array}{ccc|c|ccc}
\phantom{0} & \phantom{0} & \phantom{0} & 0 & \phantom{0} & \phantom{0} & \phantom{0} \\
 & A_1 & & \svdots & & 0 & \phantom{0} \\
 & & & a_{j,j} & & & \phantom{0} \\
 \hline
 & & & a_{j+1,j} & & & \phantom{0} \\
 & 0 & & \svdots & & A_2 \\
 & & & 0 & & & \phantom{0}
\end{array}
\right).
\]
Consequently, permuting $\sigma_1$ equidistant columns (that is, we take every $\ell$th column for an appropriate $\ell$ defined below) to the front, we obtain a $2$-stage stochastic matrix with stage sizes $(\sigma_1,S^\prime)$, where $S^\prime = \prod_{i=2}^\tau (\sigma_i + 1)$ and whose $\sigma_1+1$ diagonal blocks are again bi-diagonal, allowing us to induct on them.

Formally, let $A^{(0)} \in \Z^{S \times \sigma_1}$ be the matrix comprising the columns of $\A$ with index in $L \df \{j S^\prime: \ j = 1,\dots,\sigma_1\}$.
For $k=1,\dots,\sigma_1+1$, let
\[
A^{(k)} \in \Z^{S^\prime \times S^\prime - 1}
\]
denote the matrix where the $(i,j)$-th entry corresponds to the
$( (k-1) S^\prime + i,\, (k-1)S^\prime + j)$-th entry of $\A$.

The following $2$-stage stochastic matrix arises from $\A$ by permuting $A_0$ to the front:
\[
\left(
\begin{array}{c|}
\\
\\
A^{(0)}
\\
\\
\phantom{0}
\end{array}
\begin{array}{ccc}
\begin{tikzpicture}
\node[rectangle,draw] (r) at (0,0) {$A^{(1)}$};
\end{tikzpicture} \\
& \sddots \\
& &
\begin{tikzpicture}
\node[rectangle,draw] (r) at (0,0) {$A^{(\sigma_1+1)}$};
\end{tikzpicture}
\end{array}
\right).
\]
We have
\[
A^{(k)}_{i,j} = \A_{(k-1)S^\prime + i , (k-1)S^\prime + j} = 0
\]
whenever $i \notin \{j,j+1\}$, hence each matrix $A^{(k)}$ is again bi-diagonal.
By induction, each $A^{(k)}$ is a multi-stage matrix with $\tau - 1$ stages and stage sizes $(\sigma_2,\dots,\sigma_{\tau})$, after a suitable permutation.

$ii):$
The proof follows the same argument as for the bi-diagonal case, the only difference being that we have to delete two columns in order to split the matrix into independent blocks.

For $\tau = 1$, there is nothing to show.
For $\tau \geq 2$, define
\[
L \df \bigcup_{j=1}^{\sigma_1} \{2jS^\prime -1, 2jS^\prime\}
\]
and let $A^{(0)} \in \Z^{2S \times 2\sigma_1}$ comprise the columns with indices in $L$.
For $k=1,\dots,\sigma_1+1$, let the matrix $A^{(k)} \in \Z^{2S^\prime \times 2 S^\prime-2}$ arise from $\A$ by restricting to row indices
$2(k-1)S^\prime +1,\dots,2kS^\prime$ and column indices $2(k-1)S^\prime + 1, \dots, 2kS^\prime - 2$.
Again, the matrices~$A^{(k)}$ are tri-diagonal, and applying the induction step on them yields the claim. \qed
\end{proof}
Clearly, the result can be extended to any band-width on the diagonal. Also, the number of columns can be chosen individually even within a specific stage.
However, the simpler version above suffices for our purposes.

By considering $\A^\intercal$, we obtain the following corollary.
\begin{corollary}
\label{cor:diagonal}
Let $\tau \geq 1$, $\sigma \in \Z_{\geq 1}^\tau$, and define $S \df \prod_{i=1}^\tau (\sigma_i + 1)$.
\begin{itemize}
\item[$i)$] Let $\A^\intercal \in \Z^{ S \times S - 1}$ be bi-diagonal.
Then $\A$ is a tree-fold matrix with $\tau$ levels and level sizes $\sigma$, up to row permutations.
\item[$ii)$] Let $\A^\intercal \in \Z^{2S \times 2S - 2}$ be tri-diagonal.
Then $\A$ is a tree-fold matrix with $\tau$ levels and level sizes $2\sigma$, up to row permutations.
\end{itemize}
\end{corollary}
We close this section with a brief remark.
If $\A$ itself is bi-diagonal, we can add a zero column in the front.
This way we obtain a matrix $\tilde{A}$ for which $\tilde{A}^{\intercal}$ is bi-diagonal.
Similarly, we can add two columns in the tri-diagonal case.

\section{Multi-Stage Integer Programming}
\label{sec:MultiStage}
This section presents our main result regarding the hardness of multi-stage IPs. In particular, we reduce \textsc{3-Sat} to multi-stage IPs, proving the following:
\begin{theorem}[A lower bound for multi-stage IPs]
\label{thm:multi-stage}
For every fixed number of stages $\tau \geq 1$ and stage sizes $\sigma \in \Z^\tau_{\geq 1}$, there is no algorithm solving every instance $I$ of multi-stage integer programming in time
$2^{2^{S^{o(1)}}} |I|^{O(1)}$, where $S = \prod_{i=1}^\tau (\sigma_i + 1)$ and $|I|$ is the encoding length of $I$, unless the ETH fails.
\end{theorem}
By considering multi-stage IPs with constant stage sizes, we immediately get that every algorithm has to have a triple exponential dependency on $\tau$ when parametrized by the number of stages $\tau$ and the largest value of any coefficient of $A$.
Thus, we cannot increase the dependency on the other parameters to decrease the dependency on $\tau$.
Specifically, we rule out any algorithm solving multi-stage IPs in time less than triple exponential in $\tau$ for the parameters $\tau, \sigma, \Delta, \|c\|_{\infty}$, $\|b\|_{\infty},\|\ell\|_{\infty}$.
This proves tightness of the triple-exponential complexity w.r.t. $\tau$ of the current state-of-the-art algorithm of Klein and Reuter~\cite{klein2021collapsing}.
 
\begin{corollary}
There is a family of multi-stage integer programming instances with $\tau \geq 1$ stages,
constant stage sizes, and entries of constant value which, assuming the ETH, cannot be solved in time $2^{2^{2^{o(\tau)}}} |I|^{O(1)}$ where $|I|$ is the encoding length of the respective instance $I$.
\end{corollary}

\paragraph*{Proof idea (of Theorem~\ref{thm:multi-stage}).}
    The proof starts where the proof for the double exponential lower bound in~\cite{jansen2021double} ends. Specifically, we are given a $2$-stage stochastic matrix which is, under ETH, the double exponentially hard instance for $2$-stage stochastic IPs. 
    
    The blocks of the second stage, that is, the independent diagonal matrices are each \emph{nearly} a diagonal matrices with an extra row with (possibly) just non-zero entries. Note that this row corresponds to a summation of scaled summands, i.e., it correspond to $az_1+bz_2+cz_3+dz_4+\dots$ for row entries $a, b, \dots$ and variables $z_1, z_2, \dots$ respectively. 
    
    This sum is equal to  $s_2 + cz_3 + dz_4$ with $az_1+bz_2 = s_2$. This trick is already used in~\cite{eisenbrand2019algorithmic}, and allows us by repeated application to split the full row into an equivalent bi-diagonal matrix. Combing this with the remaining entries in the column and some re-arrangement of the rows, we get a tri-diagonal matrix for each second stage block of the original $2$-stage stochastic IP. This is the desired form to apply Lemma~\ref{lem:diagonal} to each of the second stage matrices giving us a multi-stage IP with the desired dimensions w.r.t. the ETH to proof the theorem. 

\begin{proof}[of Theorem~\ref{thm:multi-stage}]
For $x \in \Z_{\geq 0}$ and $\eta = \lceil \log(x + 1) \rceil + 1$,
let $\enc(x) \in \Z^{1 \times \eta}$ denote the binary encoding of a non-negative number $x$, i.e., $x = \sum_{i=0}^{\eta-1} 2^{i}(\enc(x))_{i+1}$ with $(\enc(x))_{i+1} \in \{0,1\}$ denoting the $i$-th bit.

In~\cite{jansen2021double}, the \textsc{$3$-Sat} problem with $N$ variables and $O(N)$ clauses is reduced to $2$-stage stochastic IPs of the form
\begin{equation}
\label{ILP}
\tag{$2$-stage stochastic IP}
\begin{pmatrix}
-e_1 & D_1 \\
-e_1 & & D_2 \\
\svdots & & & \sddots \\
-e_1 & & & & D_n
\end{pmatrix}
x =
\begin{pmatrix}
e_t \\ e_t \\ \svdots \\ e_t
\end{pmatrix},
\end{equation}
where $e_i$ is the $i$-th canonic unit vector, $t$ is the number of rows in the $D_i$ matrices being of shape

\begin{align}
\label{eq:D}
D_i &= \begin{pmatrix}
\enc(q_i) & \enc(x_i)  & \enc(y_i) \\
E &  &  \\
 & E &  \\
 &  & E \\
0 \sdots 0 & 1\,0 \sdots 0 & 1\, 0 \sdots 0
\end{pmatrix}
\text{ with }
E =
\begin{pmatrix}
2 & -1 \\
 & 2 & -1  \\
 & & \sddots & \sddots \\
 & & & 2 & -1
\end{pmatrix},
\end{align}
where $E$ is called the \emph{encoding matrix}.
Here, $x_i, y_i, q_i \in \mathbb{N}_{\geq 0}$ are numbers generated by reducing a $3$-SAT formula to $2$-stage stochastic IPs, see~\cite{jansen2021double}.
All variables have a lower bound $0$, an upper bound which is $O(2^{O(N^2\log(N))})$, and the objective function is zero, i.e., only feasibility is sought.

There are $n \in O (N^2)$ blocks $D_i \in \Z^{t \times s}$ with $s = t + 1 \in O(\log(N))$,
all coefficients are bounded by $2$ in absolute value, and the first stage size is $\sigma_1 = 1$.

For $E \in \Z^{\eta \times \eta + 1}$, let $E^\dagger$ denote the matrix arising from $E$ by reversing the order of the rows and columns, i.e., $E^\dagger_{i,j} = E_{\eta - i + 1, \eta - j + 2}$. In particular, it has the following form:
\[
E^\dagger =
\begin{pmatrix}
-1 & 2 \\
 & -1 & 2 \\
 & & \sddots & \sddots \\
 & & & -1 & 2
\end{pmatrix}.
\]
Similarly, let $\enc^\dagger (x)$ arise from $\enc(x)$ by reversing the order of the entries.
Hence, by permuting rows and columns and inserting a zero row, a single block can be brought into the shape
\[
\begin{pmatrix}
\enc(q_k) & \enc^\dagger(x_k)  & \enc(y_k) \\
E & &  \\
0 \sdots 0 & 0 \sdots 0 & 0 \sdots 0 \\
 & E^\dagger & \\
0 \sdots 0 & 0 \sdots 0\,1     & 1\,0 \sdots 0 \\
 & & E
\end{pmatrix}
\fd
\begin{pmatrix}
c^{(k)} \\
B
\end{pmatrix} \in \Z^{(t+1) \times (t+1)}
\]
where $c^{(k)} \df (\enc(q_k), \enc^\dagger(x_k), \enc(y_k))$ concatenates the bit encodings, and $B^\intercal$ is a bi-diagonal matrix, i.e.,
\begin{align}
\label{eq:B-bi-diag}
B_{i,j} &= 0 \text{ for } i \notin \{j,j-1\}.
\end{align}
In the next step, similar to the approach in~\cite{eisenbrand2019algorithmic}, we replace ${ c^{(k)} \choose B}$ by a tri-diagonal matrix $T_k$.

Fix an index $k$ and denote ${c \choose B} \df {c^{(k)} \choose B}$ for short.
Denote the global variable of~(\ref{ILP}) by $r$, the variables of the diagonal block $D_k$ by $z \df z^{(k)}$ and consider the topmost constraint
$-r + cz = 0 \Leftrightarrow \sum_{i=1}^{t+1} c_i z_i = r$ of this block.
We introduce new variables $p \in \Z^{t+1}$ and constraints as follows:
\begin{align*}
p_1 &= - r \\
p_{i+1} - p_{i} &= c_i z_i	\hspace*{3cm} i=1,\dots,t \\
-p_{t+1} &= c_{t+1} z_{t+1}.
\end{align*}
Summing up all the equations, we retrieve $\sum_{i=1}^{t+1} c_iz_i = r$.
Alternating the variables $p_i$ and $z_i$, we can replace the constraint $\sum_{i=1}^{t+1} c_i z_i = r$ with the system
\begin{align*}
\left(
\begin{array}{c|cccccccc}
1 & 1 \\
 & -1 & -c_1 &  1 \\
 & & 		 & -1 & -c_2 & 1 \\
 & & & & 				 & & \sddots \\
 & & & & & 				   & & -1 & - c_{t+1}
\end{array}
\right)
 \begin{pmatrix}
r \\ p_1 \\ z_1 \\ p_2 \\ z_2 \\ \svdots \\ z_{t+1}
\end{pmatrix}
&= \begin{pmatrix}
0 \\ \svdots \\ 0
\end{pmatrix}.
\end{align*}
We call this new system $e_1 r + \tilde{S} \tilde{z} = 0$.
Formally, we have
\[
\tilde{S}_{i,j} \df \begin{cases}
1			& j = 2i - 1, \\
-c_{i-1}	& j = 2i - 2, \\
-1			& j = 2i - 3, \\
0 			& else,
\end{cases} \hspace*{40pt}
\text{for } \left\{
\begin{array}{l}
i = 1,\dots, t+2, \\
j = 1,\dots, 2t+2.
\end{array}  \right.
\]
In particular, this definition implies
\begin{align}
\label{eq:tilde-S}
\tilde{S}_{i,j} &= 0 \quad \text{ for } \quad i \notin \{\tfrac{j+1}{2}, \tfrac{j+2}{2}, \tfrac{j+3}{2}\}.
\end{align}
We obtain $\tilde{B}$ from $B$ accordingly by adding zero columns corresponding to the new variables $p$, and adding one zero row in the top.
Formally we define
\begin{align*}
\tilde{B}_{i,j} &\df
\begin{cases}
B_{i-1,j/2} & \text{ if } i \geq 2 \text{ and } j \text{ is even,} \\
0 & \text{ else.}
\end{cases}
\qquad \text{for }
\begin{cases}
i = 1, \dots, t + 1 \\
j = 1, \dots, 2t + 2
\end{cases}
\end{align*}
By Equation~\eqref{eq:B-bi-diag}, we can observe
\begin{align}
\label{eq:tilde-B}
 \tilde{B}_{i,j} &= 0 \text{ if } i \notin \{ \tfrac{j}{2}, \tfrac{j}{2} + 1 \}.
\end{align}
So far, we reformulated $-e_1r + D z = e_t$ as an equivalent system $e_1 r + {\tilde{S} \choose \tilde{B}} \tilde{z} = e_{\tilde{t}}$ for some index $\tilde{t}$.
Finally, we can permute the rows of ${\tilde{S} \choose \tilde{B}}$, alternatingly taking a row of $\tilde{S}$ and $\tilde{B}$, and obtain the system
\begin{align*}
e_1 r + \left(
\begin{array}{cccccccc}
 1 \\
 0 & 0 & \sdots \\
 -1 & -c_1 &  1 \\
 & \star & 0 & \star \\
 & 		 & -1 & -c_2 & 1 \\
 & 		 & & \star & 0 & \star  \\
 & & & 				 & & \sddots \\
 & & & & 				   & & -1 & - c_{t+1}
 \end{array} \right) \tilde{z} = 0,
\end{align*}
where $\star \in \{-1,0,2\}$ are the corresponding entries of $E$ and $E^\dagger$ respectively.
Formally, the matrix above is defined as
\[
T_{i,j} \df \begin{cases}
\tilde{S}_{(i+1)/2,j} & i \text{ odd}, \\
\tilde{B}_{i/2,j} & i \text{ even},
\end{cases}
\]
and has dimension $(2t + 3) \times (2t + 2)$.
It remains to verify that $T$ is tri-diagonal.
Depending on $i \mod 2$ we obtain by Conditions~\eqref{eq:tilde-S} and~\eqref{eq:tilde-B}
\begin{align*}
T_{i,j} &= \tilde{S}_{(i+1)/2,j} = 0 & & \text{if } i \notin \{j,j+1,j+2\}, \\
T_{i,j} &= \tilde{B}_{i/2,j} = 0 & & \text{if } i \notin \{j, j + 2\}. \\
\end{align*}
This way, we obtain a tri-diagonal matrix $T_k \in \Z^{(2t + 3) \times (2t+2)}$ for every block ${c^{(k)} \choose B}$, and are left with the system
\begin{equation}
\begin{pmatrix}
-e_1 & T_1 \\
-e_1 & & T_2 \\
\svdots & & & \sddots \\
-e_1 & & & & T_n
\end{pmatrix}
x =
\begin{pmatrix}
e_{\tilde{t}} \\ e_{\tilde{t}} \\ \svdots \\ e_{\tilde{t}}
\end{pmatrix}
\end{equation}
that is equivalent to~\eqref{ILP}.

Choose parameters $1 \leq s \leq t + 2$ and $\tau-1 \geq 1$ s.t.\ $2(s+1)^{\tau-1} \geq 2t + 4 > 2(s+1)^{\tau-2}$, implying $2(s+1)^{o(\tau-1)} \in o(t) = o (\log(N))$.
After possibly adding some zero rows and columns we apply Lemma~\ref{lem:diagonal}, and regard each $T_k$ as a multi-stage matrix with $\tau$ stages and stage sizes $2s \cdot \mathbf{1}$;
in the worst case, i.e., choosing $\tau = 2$, the number of rows and columns are squared.
Using the global variable as a first stage,
we obtain a multi-stage integer program with $\tau$ stages and stage sizes $\sigma \df (1,2s,\dots,2s)$ that is equivalent to~\eqref{ILP}.
Furthermore, all entries are still bounded by $2$ and the dimensions are at most quadratic in the dimensions of~\eqref{ILP}.

Hence, if there is an algorithm solving every multi-stage IP in time
\[
2^{2^{ \left( \prod_{i=1}^\tau (\sigma_i + 1) \right)^{o(1)}}} \lvert I \rvert^{O(1)} \leq 2^{2^{o ( \log (N))}} \leq 2^{o(N)} \lvert I \rvert^{O(1)},
\]
we could solve \textsc{3-Sat} in time $2^{o(N)} \lvert I \rvert^{O(1)}$, contradicting the ETH~\cite{jansen2021double}.\qed
\end{proof}

\section{Tree-Fold Integer Programming}
\label{sec:TreeFold}
In this section, we consider tree-fold IPs, whose constraint matrices are the transpose of multi-stage matrices.
Our results can be viewed as a refinement of~\cite[Theorem 4]{knop2020tight}; while~\cite{knop2020tight} only considers a single parameter (namely the \emph{tree-depth}, discussed in Section~$5$), we take more aspects of the structure of the matrix into account.
For example, as one case of our lower bound, we obtain the currently best known lower bound for the special class of $n$-fold IPs, i.e., tree-fold IPs with only two levels.

We reduce from the \textsc{Subset Sum} problem.
There, we are given numbers $a_1,\dots,a_n,b \in \Z_{\geq 0}$, and the task is to decide whether there exists a vector $x \in \{0,1\}^n$ satisfying $\sum_{i=1}^n a_i x_i = b$.
Since all integers are non-negative, we can compare each $a_i$ to $b$ beforehand, and henceforth assume $0 \leq a_i < b$ for all $i$.
\begin{lemma}[{\cite[Lemma 12]{knop2020tight}}]
\label{lem:subset-sum-lb}
Unless the ETH fails, there is no algorithm for \textsc{Subset Sum} that solves every instance $a_1,\dots,a_n,b$ in time
$2^{o( n + \log(b) )}$.
\end{lemma}
\begin{theorem}[A lower bound for tree-fold IPs]
\label{thm:tree-fold}
Assuming the ETH, for every fixed $\tau \geq 1$ and $\sigma \in \Z^\tau_{\geq 1}$, there is no algorithm solving every tree-fold IP with $\tau$ levels and level sizes $\sigma$ in time
$2^{o(\prod_{i=1}^\tau (\sigma_i + 1))} |I|^{O(1)}$.
\end{theorem}
As a corollary, we obtain the following lower bound for $n$-fold IPs.

\begin{corollary}[A lower bound for $n$-fold IPs]
\label{cor:n-fold}
Assuming the ETH, there is no algorithm solving every $n$-fold IP with level sizes $(\sigma_1,\sigma_2)$ in time $2^{o(\sigma_1 \sigma_2)} |I|^{O(1)}$.
\end{corollary}

\paragraph*{Proof idea (of Theorem~\ref{thm:tree-fold}.)}
    We start with an instance of the \textsc{Subset Sum} problem. The goal is to first model this problem as an appropriate $n$-fold IP. Then, observing that the second level block-diagonal matrices are bi-diagonal allows us to apply Lemma~\ref{lem:diagonal}, yielding the desired algorithm.

    Let us thus focus on sketching how to obtain the $n$-fold matrix. We start with the straight-forward interpretation of  \textsc{Subset Sum} as an integer program, that is, $\{\sum_{i=1}^n a_i x_i = b, x \in \{0,1\}^n\}$. We could directly apply the standard encoding as in Theorem~\ref{thm:multi-stage} to lower the size of the entries, but this would yield $\sigma_1 = 1$. We could assign random rows to the first level to obtain larger values for $\sigma_1$, but this is arguably not a clean reduction for arbitrary values of $\sigma_1$. Instead, we use a doubly encoding of the large entries: First, we encode them to the base of $\Delta$ with 
     $\Delta^{\sigma_1} \geq b > \Delta^{\sigma_1-1}$. This is done similar as in the standard trick, see e.g. Theorem~\ref{thm:multi-stage}, however, we use the transposed construction. The same arguments apply. Then, the standard encoding is used to obtain the desired small entries, and the statement immediately follows. 

\begin{proof}[of Theorem~\ref{thm:tree-fold}]
For $\tau = 1$, we have arbitrary IPs, and the statement follows by~\cite{knop2020tight}.

Let $a_1,\dots,a_n, b \in \Z_{\geq 1}$ be a \textsc{Subset Sum} instance.
Choose a constant $1 \leq \sigma_1 \leq \lceil \log_2(b) \rceil$ and let $\Delta \in \Z_{\geq 2}$ s.t.\ $\Delta^{\sigma_1} \geq b > \Delta^{\sigma_1-1}$, with $\Delta = b$ if $\sigma_1=1$.
In both cases, we have
\begin{align}
\label{eq:delta-r}
\sigma_1 \log_2(\Delta) &\leq 2 \log_2 (b) \leq 2 \sigma_1 \log_2(\Delta).
\end{align}
Let $\lambda = (1,\Delta,\Delta^2,\dots,\Delta^{\sigma_1-1})^\intercal$ and let $c_i \in \{0,\dots,\Delta-1\}^{\sigma_1}$ be the $\Delta$-encoding of $a_i$, i.e., $a_i = \lambda^\intercal c_i$.
Observe that $a^\intercal x = b$ has a solution if and only if the system
\[
(C|D) {x \choose y} \df
\left(
\begin{array}{cccc|cccc}
 & & & & \Delta \\
\svdots & \svdots & & \svdots & -1 & \Delta \\
c_1 & c_2 & \sdots & c_n & & -1 & \sddots \\
\svdots & \svdots & & \svdots & & & \sddots & \Delta \\
 & & & & & & & -1
\end{array}
\right) {x \choose y}
=
\begin{pmatrix}
b \\ 0 \\ 0 \\ \svdots \\ 0
\end{pmatrix}
\]
has a solution ${x \choose y}$.
(The ``if''-direction is observing $\lambda^\intercal (C,D) = (a^\intercal,0)$;
the ``only-if''-direction follows by observing that a solution $x$ fixes the values of $y$.)

Let $t \df \lceil \log_2 \Delta \rceil$ and fix a column $c_k$ of $C$.
We can express each entry of $c_k$ in its binary representation.
This is, let $z^\intercal = (1,2,4,\dots,2^{t-1})^\intercal \in \Z^{t}$ and let $C_k \in \{0,1\}^{\sigma_1 \times t}$ be the unique matrix subject to $C_k z = c_k$.
Similarly, for the $k$-th column $d_k$, there is a unique matrix $D_k \in \{0,1\}^{\sigma_1 \times t}$ subject to $d_k = D_k z$.

Again, let the encoding matrix be the matrix
\begin{align}
\label{eq:Et-in-tree-fold}
E_t &\df
\begin{pmatrix}
2   & -1 \\
 & 2 & -1 \\
 & & \sddots & \sddots \\
 & & & 2 & -1
\end{pmatrix} \in \Z^{(t-1) \times t},
\end{align}
and observe that the system $E_t x = 0$, $0 \leq x \leq 2^t$ only has the two solutions~$\{0,z\}$.
Therefore, the \textsc{Subset Sum} instance has a solution if and only if the system
\begin{align}
\label{eq:n-fold}
\tag{$n$-fold IP}
\mathcal{A} {x \choose y} \df \begin{pmatrix}
C_1 & \sdots & C_n & D_1 & \sdots & D_{s-1} \\
E_t \\
 & E_t \\
 & & & \sddots \\
 & & & & & E_t
\end{pmatrix}
{x \choose y} =
\begin{pmatrix}
b \\ 0 \\ \svdots \\ 0
\end{pmatrix}
\end{align}
has a solution satisfying $0 \leq x \leq 2^t$.

So far, the constructed IP is an $n$-fold IP
with $\sigma_1 \in \Theta (\tfrac{\log_2(b)}{\log_2(\Delta)})$ rows in the top blocks, $t \in \Theta(\log (\Delta))$ columns per block, and $\hat{\sigma}_2 = t-1$ rows in the diagonal blocks.
We have $n+\sigma_1$ blocks in total, and $\|\A\|_{\infty} \leq 2$, hence the size of the constructed instance is polynomial in the size of the \textsc{Subset Sum} instance.
The first claim follows already:
If we could solve it in time
\[
2^{o(\sigma_1 \hat{\sigma}_2)} \leq 2^{o(\sigma_1 \log_2(\Delta))} \stackrel{\eqref{eq:delta-r}}{\leq} 2^{ o ( \log_2 (b) ) },
\]
this would contradict Lemma~\ref{lem:subset-sum-lb}, finishing the proof for $\tau = 2$ (Corollary~\ref{cor:n-fold}).

To prove the statement for $\tau \geq 3$, we continue the construction.
Choose $\tau$ s.t.\ $2 \leq \tau-1 \leq \lceil \log_2(t) \rceil$ and $s \geq 1$ s.t.\ $(s+1)^{\tau-1} \geq t > s^{\tau-1}$.
Furthermore, let $\ell \geq 0$ be s.t.\ $(s+1)^{\tau-1-\ell}s^{\ell} \geq t > (s+1)^{\tau-\ell-2}s^{\ell+1}$.
Observe that if $s=1$, then $\ell = 0$, since $\tau -1 \leq \lceil \log_2(t) \rceil$.
Let $(\sigma_2,\dots,\sigma_{\tau}) \in \{s-1,s\}^{\tau-1}$ have $\tau - 1 - \ell$ entries $s$ and $\ell$ entries $s-1$.
This careful construction of $\sigma$ allows us to estimate
\begin{align}
\label{eq:sigma-tau}
\prod_{i=2}^{\tau} (\sigma_i + 1) = (s+1)^{\tau-1-\ell}s^{\ell} &= \frac{s+1}{s} (s+1)^{\tau-\ell-2}s^{\ell+1} < 2t.
\end{align}

Since $E_t^\intercal$ is bi-diagonal, we can extend $E_t$ with zeros to a tree-fold matrix with $\tau-1$ levels and level sizes $(\sigma_2,\dots,\sigma_{\tau})$ due to Corollary~\ref{cor:diagonal}.
In total, the matrix $\A$ is a tree-fold matrix with $\tau$ levels and level sizes~$\sigma$, and its encoding length did not change.
If there is an algorithm solving every such tree-fold IP in time

\[
2^{o( (\sigma_1+1) \prod_{i=2}^{\tau} (\sigma_i+1) )}
\stackrel{\eqref{eq:sigma-tau}}{\leq} 2^{o( 2\sigma_1 \cdot 2t )}
\leq 2^{o( 8 \sigma_1 \log_2(\Delta) )}
\stackrel{\eqref{eq:delta-r}}{\leq} 2^{o( 16 \log_2(b))},
\]
this would again contradict Lemma~\ref{lem:subset-sum-lb}. \qed
\end{proof}

\section{Lower Bounds On The Graver Basis}
\label{sec:Graver}

\begin{definition}
Let $A \in \Z^{m \times n}$.
The Graver basis $\mathcal{G}(A)$ of $A$ is the set of all vectors $z \in \ker(A) \cap \Z^n\setminus{(0, \dots, 0)}$ that cannot be written as $z = x + y$ with $x,y \in \ker(A) \cap \Z^n$ satisfying $x_iy_i \geq 0$ for all $i$.
For any norm $K$, we define $g_{K} (A) = \max_{g \in \mathcal{G}(A)} \| g \|_{K}$.
\end{definition}

Our previous results state that, assuming ETH, there is no algorithm for multi-stage or tree-fold IPs that solves \emph{every}
instance within a certain time threshold.
We now show further evidence orthogonal to the ETH.
All known algorithms for block-structured IPs have complexities which are at least $g_\infty(A)$ or $g_1(A)$ for multi-stage IPs or tree-fold IPs, respectively.
Thus, if there is no fundamentally different algorithm for those problems, lower bounding $g_\infty(A)$ and $g_1(A)$ lower bounds the complexity of those problems.
We show that the instances we constructed, in particular, the encoding matrix $E_t$ (see~\eqref{ILP} and~\eqref{eq:n-fold}), indeed have large $g_\infty(A)$ and $g_1(A)$, respectively.

\begin{lemma}
\label{lem:graver}
Let $t \geq 2$, $\Delta \in \Z_{\geq 2}$.
The encoding matrix
\[
E_t(\Delta) \df
\begin{pmatrix}
\Delta   & -1   \\
 & \Delta & -1  \\
 & & \sddots  & \sddots \\
 & & & \Delta & -1
\end{pmatrix} \in \Z^{(t-1) \times t}
\]
satisfies $g_{\infty}(E_t(\Delta)) \geq \Delta^{t-1}$ and $g_1(E_t(\Delta)) \geq \tfrac{\Delta^t - 1}{\Delta - 1}$.
\end{lemma}

\begin{proof}
The matrix $E_t(\Delta)$ has full row rank.
Therefore, its kernel has rank $t-(t-1) = 1$.
Since $z^\intercal \df (1,\Delta,\dots,\Delta^{t-1}) \in \ker(E_t)$, every element in $\ker(E_t) \cap \Z^{t}$ is an integer multiple of $z$.
Thus, $\{z, -z\}$ is the Graver basis of $E_t(\Delta)$. \qed
\end{proof}
\cite[Lemma 2]{eisenbrand2019proximity} states that $g_1(A) \leq (2 m \| A \|_\infty + 1)^m$ for any matrix $A \in \Z^{m \times n}$, so $E_t$ is almost the worst case.
As an immediate consequence of Lemma~\ref{lem:diagonal}, we get:
\begin{theorem}
Let $\tau \in \Z_{\geq 1}$, $\sigma \in \Z^{\tau}_{\geq 1}$, and define $S \df \prod_{i=1}^\tau (\sigma_i + 1)$.
\begin{enumerate}
\item There is a multi-stage matrix $A$ with $\tau$ stages and stage sizes $\sigma$ that satisfies
$g_{\infty}(A) \geq \| A \|_{\infty}^{S - 2}$.
\item There is a tree-fold matrix $A$ with $\tau$ levels and level sizes $\sigma$  that satisfies
$g_{1}(A) \geq \| A \|_{\infty}^{S - 1}$.
\end{enumerate}
\end{theorem}
\begin{proof}
For the first claim, observe that the matrix $E_{S - 1}(\Delta)$ satisfies $g_{\infty}(E_{S - 1}(\Delta)) \geq \Delta^{S-2}$.
If $A$ arises from $E_{S - 1}(\Delta)$ by adding a zero row in the top and the bottom, we immediately obtain $g_{\infty}(A) \geq \Delta^{S-2}$, and can apply Lemma~\ref{lem:diagonal}.

For the second claim, we can apply Corollary~\ref{cor:diagonal} to the matrix $E_{S}(\Delta)$. \qed
\end{proof}

\section{Beyond Block Structure: Tree-Depth}
\label{sec:tree-depth}
There is the more general notion of primal and dual \emph{tree-depth} of $A$ to capture the above classes of block-structured IPs.
This section provides a brief introduction and states our results in terms of these parameters.

The \emph{primal graph} $G$ of $A$ has the columns of $A$ as a vertex set, and an edge between two columns $a,b$ if they share a non-zero entry, i.e., there is an index $i$ with $a_ib_i \neq 0$.
A \emph{td-decomposition} of $G$ is an arborescence $T$ with $V(T) = V(G)$ s.t.\ for any edge $\{u,v\} \in E(G)$ there either is a $(u,v)$-path in $T$ or a $(v,u)$-path.
A td-decomposition of $G$ with minimum height is a minimum td-decomposition.
The primal tree-depth $\td_P(A)$ of $A$ is the maximum number of vertices on any path in a miminum td-decomposition of $G$.
The \emph{dual graph} and tree-depth are the primal graph and tree-depth of $A^\intercal$.

If $A$ is a multi-stage matrix with $\tau$ stages and stage sizes $\sigma$, we can construct a primal td-decomposition:
Start with a path through the first $\sigma_1$ columns.
Since the rest of the matrix decomposes into blocks, there are no edges between any two blocks, and any edge inducing a path either is from one of the first $\sigma_1$ columns to a block, or within a block.
Thus, we can recurse on the blocks and append another path of $\sigma_2$ columns to column $\sigma_1$.
This way, we obtain a td-decomposition of height $\sum_{i=1}^\tau \sigma_i - 1$, and get:
\begin{lemma}
\leavevmode
\begin{enumerate}
\item Let $A$ be multi-stage with $\tau$ stages and stage sizes $\sigma$.
Then $\td_P(A) \leq \sum_{i=1}^\tau \sigma_i$.
\item Let $A$ be tree-fold with $\tau$ levels and level sizes $\sigma$.
Then $\td_D(A) \leq \sum_{i=1}^\tau \sigma_i$.
\end{enumerate}
\end{lemma}
Using the constructions of Theorems~\ref{thm:multi-stage} and~\ref{thm:tree-fold}, we have $\sigma \leq 2 \cdot \mathbf{1}$ and hence, $\td_P(\A) \leq 2 \tau$, $\td_D(\A) \leq 2\tau$ respectively.
We obtain the following corollary.
While the second point was already proven in~\cite{knop2020tight}, the first point is a new consequence.
\begin{corollary}
Assuming the ETH,
there is no algorithm solving every IP in time
$2^{2^{2^{o( \td_P(A))}}} |I|^{O(1)}$,
nor in time
$2^{o (2^{\td_D(A)})} |I|^{O(1)}$.
\end{corollary}

\bibliography{ref}

\end{document}